\documentclass{IEEEoj}
\usepackage{cite}
\usepackage{amsmath,amssymb,amsfonts,comment,bm}
\usepackage{algorithmic}
\usepackage{tikz}
\usetikzlibrary{arrows,decorations.markings,automata}
\usepackage{amsthm}
\usepackage{mathrsfs}
\usepackage{pgfplots}
\pgfplotsset{compat=1.15}
\usepackage{mathrsfs}
\usetikzlibrary{arrows}
\usepackage{rotating}
\usepackage{graphicx,color,float}
\usepackage{extarrows}
\usepackage{textcomp}
\newtheorem{theorem}{Theorem}

\newtheorem{corollary}{Corollary}
\newtheorem{lemma}{Lemma}

\def\I{\mathcal{I}}
\def\RR{\mathcal{R}}

\def\I{\mathcal{I}}

\def\C{\mathbb{C}}

\def\x{\boldsymbol{x}}

\def\h{\mathbf{h}}

\def\BibTeX{{\rm B\kern-.05em{\sc i\kern-.025em b}\kern-.08em
    T\kern-.1667em\lower.7ex\hbox{E}\kern-.125emX}}
\AtBeginDocument{\definecolor{ojcolor}{cmyk}{0.93,0.59,0.15,0.02}}
\def\OJlogo{\vspace{-14pt}\includegraphics[height=28pt]{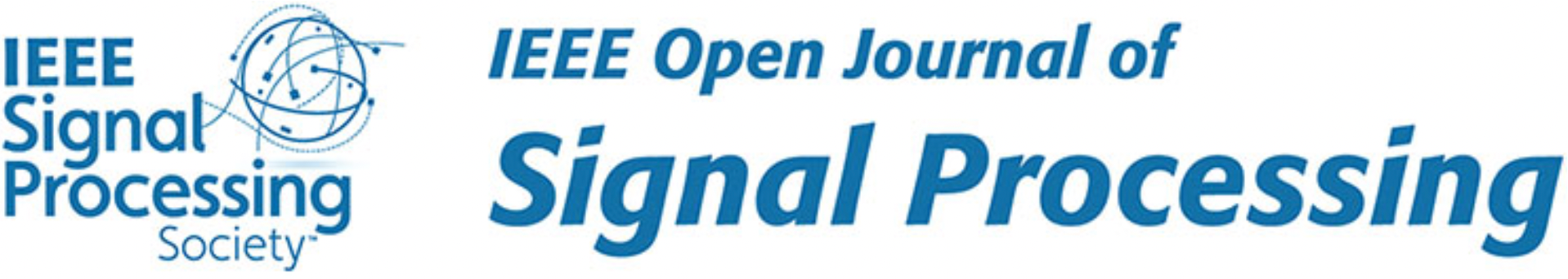}}
\begin{document}
\receiveddate{26 October, 2022}
\reviseddate{7 January, 2023}
\accepteddate{7 January, 2023}
\publisheddate{XX Month, XXXX}
\currentdate{XX Month, XXXX}
\title{Diversity Order Analysis for Quantized Constant Envelope Transmission}

\author{Zheyu Wu$^{1,2}$, Jiageng Wu$^{3}$, Wei-Kun Chen$^{4}$, AND 
Ya-Feng Liu$^{1}$}
\affil{LSEC, ICMSEC, AMSS, Chinese Academy of Sciences, Beijing, China}
\affil{School of Mathematical Sciences, University of Chinese Academy of Sciences, Beijing, China}
\affil{School of Mathematics, Jilin University, Changchun, China}
\affil{School of Mathematics and Statistics, Beijing Institute of Technology, Beijing, China}
\corresp{CORRESPONDING AUTHOR: Ya-Feng Liu (e-mail: yafliu@lsec.cc.ac.cn).}
\authornote{The work of Zheyu Wu and Ya-Feng Liu was supported in part by the National Natural Science Foundation of China (NSFC) under Grant 12288201, Grant 12022116, and Grant 11991021.
 The work of Wei-Kun Chen was supported in part by the NSFC under Grant 12101048 and Beijing Institute of Technology Research Fund Program for Young Scholars.}
\markboth{Diversity Order Analysis for Quantized Constant Envelope Transmission}{Wu \textit{et al.}}

\begin{abstract}
Quantized constant envelope (QCE) transmission is a popular and effective technique to reduce the hardware cost and improve the power efficiency of 5G and beyond systems equipped with large antenna arrays.
It has been widely observed that the number of quantization levels has a substantial impact on the system performance. This paper aims  to quantify the impact of the number of quantization levels on the system performance. Specifically, we consider a downlink single-user multiple-input-single-output (MISO) system with $M$-phase shift keying (PSK) constellation under the Rayleigh fading channel.  We first derive a novel bound on the system symbol error probability (SEP). Based on the derived SEP bound, we characterize the achievable diversity order of the quantized matched filter (MF) precoding  strategy. Our results show that full diversity order can be achieved when the number of quantization levels $L$ is greater than the PSK constellation order $M,$ i.e.,  $L>M$, only  half diversity order is achievable when $L=M$,  and the achievable diversity order is $0$ when $L<M$. Simulation results verify our theoretical analysis.
\end{abstract}

\begin{IEEEkeywords}
Diversity order analysis, Large antenna array, QCE transmission,  SEP.
\end{IEEEkeywords}


\maketitle

\section{INTRODUCTION}
\IEEEPARstart{L}{arge} antenna array is a promising technology to achieve high data rate and high reliability of wireless communication systems \cite{massivemimo2,massivemimo1,massivemimo3}. However, the power consumption and hardware cost of the system also grow with the number of antennas, which is a major concern for the practical implementation of the large antenna array technology. To address such issues, it is necessary to employ low-cost and energy-efficient hardware components at the base station (BS). 
It is well known that the most power hungry components at the BS  are the  power amplifiers (PAs) \cite{pa}. 
To improve  the efficiency of the PAs, transmit signals with low peak-to-average power ratios (PAPRs) are desirable. In particular, constant envelope (CE) transmission, where the transmit signals from each antenna are restricted to have the same amplitude, has attracted a lot of  research interests as it facilitates the use of the most efficient and cheapest PAs.
It has been shown in the pioneering works \cite{CE2,CE} that with $N$ transmit antennas at the BS, an ${O}(N)$ array power gain is achievable for CE transmission, as in the case of conventional transmission schemes without the CE constraint. In addition, numerous well-designed CE transmission strategies have been proposed and have been shown (via simulations) to have good symbol error rate (SER) performance, see, e.g.,  \cite{CEdesign1,CEdesign2,CEdesign3,CEdesign4,CEdesign5} and the references therein. {\color{black}CE transmission has also found wide applications in many other scenarios \cite{phy,radar,irs}.}

However, a practical issue associated with  CE transmission is that  the digital-to-analog converters (DACs) at each antenna element must have infinite or very high resolution to ensure that the transmit signals can take any phase. This will lead to high hardware cost and power consumption of the communication system since high-resolution DACs are expensive and the power consumption of the DACs increases exponentially with the resolution number \cite{resolution1}. Due to this, a more practical transmission scheme called quantized constant envelope (QCE) transmission \cite{trellis} has been considered recently, where low-resolution DACs are employed and the phases of the transmit signals can only be selected from a (possibly small) finite set. 

Existing works on QCE transmission mainly focused on precoding design \cite{trellis,ciqce,AM,GEMM,IDE}, and in particular one-bit precoding design \cite{SQUID,C3PO2,SEP2,CIfirst,CImodel,PBB,conference, journal}, which is a special case of QCE transmission. The only few works that considered the performance analysis all focused on the one-bit case. 
Specifically, the authors in \cite{SQUID,MFrate} derived lower bounds on the achievable rate of a one-bit MIMO system.  The result in \cite{MFrate} was further extended to the frequency selective channel in \cite{frequency}. In \cite{ZF}, the authors considered the one-bit zero-forcing precoder and derived closed-form symbol error probability (SEP) approximations for large antenna array systems. To the best of our knowledge, there is still a theoretical gap in the performance analysis of general QCE transmission.

The motivation behind this work is based on the following observations that have been drawn from the simulations in existing works.
First, the CE transmission can generally achieve good SER performance \cite{CE} while one-bit transmission sometimes suffers from a severe SER floor \cite{SQUID}.  Second,  slightly increasing the resolution of DACs from $1$ bit to $2-3$ bits can sometimes significantly
improve the SER performance of the system \cite{IDE}. The goal of this paper is to theoretically characterize the system performance of a simple but popularly used QCE transmission strategy and shed some light on the above observations.

In this paper, we consider a downlink single-user MISO system with $M$-phase shift keying (PSK) modulation. The main contributions of this paper are twofold. First, we derive a new bound on the system SEP and the bound only involves an elegant quantity known as the safety margin \cite{CI3,CI2,CItutorial}. The bound is universal and independently interesting, because it might be useful for the SEP analysis of possibly many communication scenarios. Second and more importantly, we characterize the diversity order of the quantized matched filter (MF) precoder. {\color{black}The reasons for the choice of the quantized MF precoder is that it is asymptotically optimal when the number of quantization levels goes to infinity in the considered system, and it admits a closed-form expression and thus is amenable to analysis.} We show that the quantized MF precoder is able to achieve full diversity order when the number of quantization levels $L$ is larger than $M,$ while it can only achieve half and zero diversity order when $L=M$ and $L<M$, respectively. {\color{black}The above results hold as long as $L$ and $M$ are positive integers and $M>1$
}. Simulation results show that the analysis results are correct. 

The remaining parts of this paper are organized as follows. Section \ref{section2} describes the system model and the problem formulation. Section \ref{section3} derives an important inequality on the SEP, which serves as the main tool for our analysis. Section \ref{section4} gives the diversity order analysis.  Simulations are given in Section \ref{section5} to verify our theoretical results and the paper is concluded in Section \ref{section6}.

Throughout the paper,  we use  $x$, $\x$, and $\mathcal{X}$ to denote scalar, vector, and set, respectively. For a scalar $x\in\C$, $|x|$, $\arg(x)$, $\RR(x)$, and $\I(x)$ return the absolute value, the argument, the real part, and the imaginary part of $x$, respectively. For a vector $\x\in\C^n$, $x_i$ denotes the $i$-th entry of $\x$ and $\|\x\|_p$ denotes the $p$-norm of $\x$, where $p\in\{1,2\}$; $\x^\mathsf{T}$, $\x^\dagger$, and $\x^{\mathsf{H}}$ denote the transpose, the conjugate, and the Hermitian transpose of $\x$, respectively.  For a random variable $X$, we use $p_X(\cdot)$ and $F_X(\cdot)$ to denote its probability density function (PDF) and cumulative distribution function (CDF), respectively. $\mathbb{E}\left[\cdot\right]$ and $\mathbb{P}(\cdot)$ return the expectation and the probability of their corresponding argument, respectively. $\mathcal{CN}(0,\sigma^2)$ and $\mathcal{N}(0,\sigma^2)$ represent the zero-mean circularly symmetric complex Gaussian distribution and zero-mean Gaussian distribution (in the real space) with variance $\sigma^2$, respectively. Finally, $j$ denotes the imaginary unit (satisfying $j^2=-1$).

\vspace{-0.2cm}
\section{SYSTEM MODEL AND PROBLEM FORMULATION}\label{section2}
Consider a MISO system where an $N$-antenna BS serves a single-antenna user.  Let $\h=(h_1,h_2,\dots,h_N)^\mathsf{T}$ denote the channel vector between the BS and the user and $\x=(x_1,x_2,\dots,x_N)^\mathsf{T}$ denote the transmitted signal from the BS. The signal received by the user can then be expressed as 
\begin{equation*}
y=\h^\mathsf{T}\x+n,
\end{equation*}
where $n\sim\mathcal{C}\mathcal{N}(0,\sigma^2)$ is the additive white Gaussian noise.
We  assume that the entries of $\h$ are independent and identically distributed (i.i.d.) following $\mathcal{C}\mathcal{N}(0,1)$. 

In this paper, we consider  QCE transmission, i.e.,  each antenna is only allowed to transmit QCE signals. Mathematically,  the QCE constraint can be expressed as $$x_i\hspace{-0.05cm}\in\hspace{-0.05cm}\mathcal{X}_L\hspace{-0.03cm}\triangleq\left\{\hspace{-0.06cm}\sqrt{\frac{P_T}{N}}e^{j\frac{(2l-1)\pi}{L}}\hspace{-0.15cm},~l\hspace{-0.03cm}=\hspace{-0.03cm}1,2,\dots,L\right\}\hspace{-0.1cm},~i=1,2,\dots,N,$$ where $P_T$ is the total transmit power at the BS and  $L$ is the number of quantization levels, {\color{black}i.e., the number of points in $\mathcal{X}_L$}. In what follows, we set $P_T=1$ for simplicity. 
 Let $s$ be the intended symbol for the user, which is independent of $\h$ and $n$. We focus on $M$-PSK constellation, i.e., $s$ is uniformly drawn from  $\mathcal{S}_M=\{e^{j\frac{2\pi (m-1)}{M}},~m=1,2,\dots,M\}$.  
{\color{black}  We note here that in practice, both $L$ and $M$ should be a power of $2$, and an $L$-level quantization corresponds to $(\log_2 {L}-1)$-bit quantization of DACs.  However, since our analysis holds for arbitrary quantization level $L$ and constellation order $M>1$, we choose to present our analysis and results in the most general form and just assume that $L$ and $M$ are positive integers and $M>1$ in the following.} 

In this paper, we adopt a simple QCE transmission strategy called quantized MF precoding:
{\color{black}\begin{equation}\label{MF}
\x=q_L(s\h^\dagger),
\end{equation} where $q_L(\cdot)$ is the quantization function that maps its argument component-wise to the nearest points in $\mathcal{X}_L$. The reasons for the choice of the quantized MF precoder are as follows. First, it is asymptotically optimal when the number of quantization levels $L$ goes to infinity. More specifically, when there is no quantization (i.e., $L=\infty$), the transmitted signals of quantized MF precoding can be expressed as 
\begin{equation}\label{qmf}
 x_i=\frac{1}{\sqrt{N}}e^{-j\arg(h_i)}s,~i=1,2,\dots,N, 
 \end{equation} with which the useful signal power is maximized. Second, the quantized MF precoder admits a closed-form expression (i.e., \eqref{MF}) and is of low computational complexity, which is more practical and more amenable to analysis than many other transmission strategies depending on numerical solutions of certain discrete optimization problems (e.g., those in \cite{ciqce,AM,GEMM,IDE}). With quantized MF precoding, the system model can be expressed as }
\begin{equation}\label{DCEsys}
y=\h^\mathsf{T}q_L(s\h^\dagger)+n=\frac{1}{\sqrt{N}}\sum_{i=1}^N|h_i|e^{j\theta_i}s+n,
\end{equation}
where $\theta_i\in[-\frac{\pi}{L},\frac{\pi}{L}]$, $i=1,2,\dots, N$, are the quantization error (of angle).  At the receiver side, we assume that nearest neighbor decoding is employed, that is, the user maps its received signal $y$ to the nearest constellation point $\hat{s}\in\mathcal{S}_M$. 

Our goal in this paper is to characterize the impact of the number of quantization levels on the  system performance. We adopt the diversity order \cite{diversity} as our performance metric\footnote{{\color{black}There are also many other important performance metrics, including the spectral and energy efficiency. Investigating the impact of the number of quantization levels under these performance metrics can be interesting future works.}}, which is a classical metric that characterizes the rate at which the SEP, i.e., $\mathbb{P}\left(\hat{s}\neq s\right)$, tends to zero as the signal-to-noise-ratio (SNR) grows. Its  definition is given by
\begin{equation}\label{def:diversity}
d=\lim_{\rho\to+\infty}-\frac{\log \mathbb{P}(\hat{s}\neq s)}{\log\rho},
\end{equation}
where $\rho=\frac{1}{\sigma^2}$ is the SNR.
Roughly speaking,  the above definition says that in the high SNR regime, the SEP will scale as $\rho^{-d}$. 

\vspace{-0.2cm}
\section{A NEW BOUND ON THE SEP}\label{section3}
In this section, we derive a new bound on the system SEP. The new bound will serve as the main tool for our subsequent analysis and is independently interesting, because it might also be useful for the SEP analysis of other  communication scenarios.

Consider the following system model:
\begin{equation}\label{sys}
y=\beta s+n,
\end{equation}
where $s$ is uniformly drawn from $M$-PSK constellation, $n\sim\mathcal{C}\mathcal{N}(0,\sigma^2)$, and $\beta\in \C$ is a constant. Define 
\begin{equation}\label{alpha}
\alpha=\RR(\beta)-|\I(\beta)|\cot\frac{\pi}{M}.
\end{equation}
The above quantity is well-known as the safety margin in the literature and is widely adopted as a  performance metric for precoding design \cite{CI3,CI2,CItutorial}. Geometrically, it characterizes the distance between the noise-free received signal and the decision boundary of the intended symbol, as shown in Fig. \ref{fig1}. It is widely believed that the safety margin has an essential impact on the system SEP.  
In the following theorem, we give a quantitative characterization of the relationship between the SEP and the safety margin $\alpha$. 




\begin{theorem}\label{sepbound}
The {\normalfont{SEP}} of model \eqref{sys} can be bounded as
\begin{equation}\label{in3}
Q\left(\frac{\sqrt{2}\sin\frac{\pi}{M}\alpha}{\sigma}\right)\leq\text{\normalfont{SEP}}\leq 2Q\left(\frac{\sqrt{2}\sin\frac{\pi}{M}\alpha}{\sigma} \right),
\end{equation}
where $\alpha$ is given in \eqref{alpha} and $Q(x)=\frac{1}{\sqrt{2\pi}}\int_x^{\infty}e^{-\frac{1}{2}x^2}dx$ is the tail distribution function of the standard Gaussian distribution.
\end{theorem}\vspace{-0.4cm}
\begin{proof}
Let $\mathcal{D}_s$ denote the decision region of symbol $s$ and let $s_1,s_2, \dots, s_M$ be the constellation points in $\mathcal{S}_M$, where $s_m=e^{\frac{j2\pi(m-1)}{M}}$.  Then, the SEP can be expressed as 
\begin{equation}\label{eq:proof5}
\begin{aligned}
\text{SEP}=\mathbb{P}\left(y\notin \mathcal{D}_s\right)
=\frac{1}{M}\sum_{m=1}^M\mathbb{P}\left(y\notin \mathcal{D}_s|s=s_m\right),
\end{aligned}
\end{equation}
where the second equality holds since $s$ is uniformly drawn from $\mathcal{S}_M$. Due to the symmetry of the PSK constellation, each probability in the above summation is equal and thus we only need to consider one of them. For simplicity, we consider $\mathbb{P}(y\notin \mathcal{D}_s|s=s_1)$,  where $s_1=1$. Since nearest neighbor decoding is employed, the decision region of $s_1$ is
$$\begin{aligned}
\mathcal{D}_{s_1}&=\left\{y\mid\arg(y)\in\left(-\frac{\pi}{M},\frac{\pi}{M}\right)\right\}\\
&=\left\{y\mid |\I(y)|< \tan\frac{\pi}{M}\RR(y)\right\};
\end{aligned}$$
see Fig. \ref{fig1} for an illustration of the decision region of $s_1$. Based on this, we can express $\mathbb{P}(y\notin D_s|s=s_1)$ as

\definecolor{ccqqqq}{rgb}{0.8,0,0}
\definecolor{qqqqtt}{rgb}{0,0,0.2}
\definecolor{yqyqyq}{rgb}{0.5019607843137255,0.5019607843137255,0.5019607843137255}
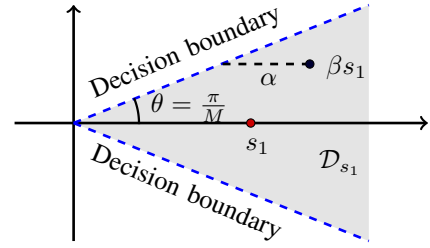
\begin{figure}
\centering
\begin{tikzpicture}[scale=0.785]
\fill[line width=0pt,color=yqyqyq,fill=yqyqyq,fill opacity=0.21] (-1,2) -- (-6,0) -- (-1,-2.0) -- cycle;
\draw [->,line width=1pt] (-7,0) -- (0,0);
\draw [->,line width=1pt] (-6,-2) -- (-6,2);
\draw [dashed,line width=1pt,blue] (-6,0) -- (-1,2);
\draw [dashed,line width=1pt,blue] (-6,0) -- (-1,-2);
\draw [shift={(-6,-0)},line width=0.85pt]  plot[domain=-0.017192971048285877:0.39,variable=\t]({1*1.107172122604702*cos(\t r)+0*1.107172122604702*sin(\t r)},{0*1.107172122604702*cos(\t r)+1*1.107172122604702*sin(\t r)});
\draw (-4.85,0.65) node[anchor=north west] {$\theta=\frac{\pi}{M}$};
\draw (-3.25,-0.1) node[anchor=north west] {$s_1$};
\draw (-3,1) node[anchor=north west] {$\alpha$};
\draw (-2,-0.3) node[anchor=north west] {$\mathcal{D}_{s_1}$};
\draw (-1.88,1.3) node[anchor=north west] {$\beta s_1$};
\draw [dashed,line width=1pt] (-2,1) -- (-3.5,1);
\draw [fill=qqqqtt] (-2,1) circle (2pt);
\draw [fill=ccqqqq] (-3,0) circle (2pt);
\begin{turn}{23} 
\draw (-3.3,2.7) node {Decision boundary};
\end{turn}
\begin{turn}{-23} 
\draw (-3.3,-2.7) node{Decision boundary};
\end{turn}
\end{tikzpicture}
\vspace{-0.6cm}
\caption{An illustration of the decision region $\mathcal{D}_{s_1}$ of symbol $s_1$ and  its safety margin $\alpha$. \label{fig1}}
\end{figure}
\begin{equation*}\label{eq:proof1}
\begin{aligned}
~\mathbb{P}(y\notin \mathcal{D}_s|s=s_1)&=\mathbb{P}(\beta+n\notin \mathcal{D}_{s_1})\\
&=\mathbb{P}\left(|\I(\beta+n)|\geq \tan\frac{\pi}{M}\RR(\beta+n)\right),
\end{aligned}
\end{equation*}
which can further be lower and upper bounded by
\begin{equation}\label{eq:proof2}
\begin{aligned}
 &\quad~\max\left\{\mathbb{P}\left(\I(\beta+n)\geq \tan\frac{\pi}{M}\RR(\beta+n)\right),\right.\\
 &\qquad\hspace{0.7cm} \left.\mathbb{P}\left(\I(\beta+n)\leq -\tan\frac{\pi}{M}\RR(\beta+n)\right)\right\}\\
&\leq\mathbb{P}\left(|\I(\beta+n)|\geq \tan\frac{\pi}{M}\RR(\beta+n)\right)\\
&\leq\mathbb{P}\left(\I(\beta+n)\geq \tan\frac{\pi}{M}\RR(\beta+n)\right)\\
&\quad+\mathbb{P}\left(\I(\beta+n)\leq -\tan\frac{\pi}{M}\RR(\beta+n)\right).
\end{aligned}
\end{equation}
For $\mathbb{P}\left(\I(\beta+n)\geq \tan\frac{\pi}{M}\RR(\beta+n)\right)$, we have
\begin{equation}\label{eq:proof3}
\begin{aligned}
&\quad~\mathbb{P}\left(\I(\beta+n)\geq \tan\frac{\pi}{M}\RR(\beta+n)\right)\\
&=\mathbb{P}\left(\cot\frac{\pi}{M}\I(n)-\RR(n)\geq \RR(\beta)-\cot\frac{\pi}{M}\I(\beta)\right)\\
&=Q\left(\frac{\sqrt{2}\sin\frac{\pi}{M}}{\sigma}\left(\RR(\beta)-\cot\frac{\pi}{M}\I(\beta)\right)\right),
\end{aligned}
\end{equation}
where the last equality holds since $n\sim\mathcal{CN}(0,\sigma^2)$ and thus $\cot\frac{\pi}{M}\I(n)-\RR(n)\sim\mathcal{N}\left(0,\frac{\sigma^2}{2\sin^2\frac{\pi}{M}}\right)$. Similarly, we can show that 
\begin{equation}\label{eq:proof4}
\begin{aligned}
&\quad~\mathbb{P}\left(\I(\beta+n)\leq -\tan\frac{\pi}{M}\RR(\beta+n)\right)\\
&=Q\left(\frac{\sqrt{2}\sin\frac{\pi}{M}}{\sigma}\left(\RR(\beta)+\cot\frac{\pi}{M}\I(\beta)\right)\right).
\end{aligned}
\end{equation}
Note that $\alpha=\min\{\RR(\beta)-\cot\frac{\pi}{M}\I(\beta),\RR(\beta)+\cot\frac{\pi}{M}\I(\beta)\}$ and $Q(\cdot)$ is a decreasing function. Combining this with \eqref{eq:proof5}--\eqref{eq:proof4}, we have  the desired result in \eqref{in3}, which completes the proof.


\end{proof}
Several remarks on Theorem \ref{sepbound} are in order. First, the upper bound in \eqref{in3} has already been derived in \cite{SEP3}.  Our proof is different from that in \cite{SEP3} and can give both lower and upper bounds in \eqref{in3} simultaneously. It turns out that both the lower and upper bounds in \eqref{in3} are important in the following diversity order analysis.  Second, 
if $\beta$ in \eqref{sys} is a positive constant, i.e., $\beta>0$, then $\alpha=\beta$ and inequality  \eqref{in3} reduces to the following well-known inequality \cite{digitalcommunication}:
\begin{equation*}\label{beta>0}
Q\left(\frac{\sqrt{2}\sin\frac{\pi}{M}\beta}{\sigma}\right)\leq\text{\normalfont{SEP}}\leq 2Q\left(\frac{\sqrt{2}\sin\frac{\pi}{M}\beta}{\sigma} \right).
\end{equation*}
Finally, Theorem \ref{sepbound} enables to characterize the SEP of the considered model \eqref{DCEsys}, as shown in the following corollary.
 \begin{corollary}\label{sepbound2}
 The {\normalfont{SEP}} of model \eqref{DCEsys} can be bounded as 
 \begin{equation}\label{in4}
\hspace{-0.2cm}\mathbb{E}_\alpha\hspace{-0.1cm}\left[Q\hspace{-0.05cm}\left(\hspace{-0.05cm}\frac{\sqrt{2}\sin\frac{\pi}{M}\alpha}{\sigma}\right)\right]\hspace{-0.1cm}\leq\hspace{-0.02cm}\text{\normalfont{SEP}}\hspace{-0.02cm}\leq \hspace{-0.03cm}2\mathbb{E}_\alpha\hspace{-0.1cm}\left[Q\hspace{-0.05cm}\left(\hspace{-0.05cm}\frac{\sqrt{2}\sin\frac{\pi}{M}\alpha}{\sigma}\hspace{-0.05cm}\right)\hspace{-0.05cm}\right],
\end{equation}
where 
\begin{equation}\label{alpha2}
\alpha=\frac{1}{\sqrt{N}}\sum_{i=1}^N|h_i|\left(\cos\theta_i-|\sin\theta_i|\cot\frac{\pi}{M}\right).
\end{equation}
 \end{corollary}
 \begin{proof}
The considered model \eqref{DCEsys} is in the form of \eqref{sys} with 
\begin{equation}\label{beta}
\beta=\frac{1}{\sqrt{N}}\sum_{i=1}^N|h_i|e^{j\theta_i}.
\end{equation}
Note that $\beta$ in \eqref{beta} is a random variable, which is a function of both $\h$ and $\theta_i, \, i=1,2,\dots, N$.  Applying the total probability theorem, we can express the SEP  of model \eqref{DCEsys} as 
\begin{equation}\label{tpt}\begin{aligned}
\text{SEP}=\mathbb{P}(\hat{s}\neq s)&=\int_{\beta\in\C}\mathbb{P}\left(\hat{s}\neq s\mid\beta=x\right)p_\beta(x)dx\\
&=\mathbb{E}_{\beta}\left(\mathbb{P}(\hat{s}\neq s\mid\beta)\right).
\end{aligned}
\end{equation} One can show that $\beta$ and $s$ are independent. Then it follows  from \eqref{tpt} and Theorem \ref{sepbound} that 
\begin{equation*}
\mathbb{E}_\alpha\hspace{-0.1cm}\left[Q\hspace{-0.05cm}\left(\hspace{-0.05cm}\frac{\sqrt{2}\sin\frac{\pi}{M}\alpha}{\sigma}\right)\right]\leq\text{SEP}\leq 2\mathbb{E}_\alpha\hspace{-0.1cm}\left[Q\hspace{-0.05cm}\left(\hspace{-0.05cm}\frac{\sqrt{2}\sin\frac{\pi}{M}\alpha}{\sigma}\hspace{-0.05cm}\right)\hspace{-0.05cm}\right],
\end{equation*}
where $\alpha$ is given in \eqref{alpha2}.

\end{proof}

In the following section, we will use Corollary \ref{sepbound2} as the main tool  to analyze the  diversity order of the considered system. 


\section{DIVERSITY ORDER ANALYSIS}\label{section4}
In this section, we will first present our main diversity order results and give some explanations and discussions in Section \ref{section4}-A and then give a detailed proof in Section \ref{section4}-B.
\subsection{MAIN RESULTS}\label{sec:4a}
We first summarize the diversity order results in the following theorem. 
\begin{theorem}\label{diversity}
For a single-user MISO system with $N$ transmit antennas and M-PSK modulation, the achievable diversity order of the $L$-level quantized MF precoder in \eqref{MF} is given by 
\begin{equation*}
d=\left\{
\begin{aligned}
N,\quad \text{if }L&> M;\\
\frac{N}{2},\quad\text{if }L&=M;\\
0,~\quad \text{if }L&< M.
\end{aligned}\right.
\end{equation*}
\end{theorem}


Theorem 2 clearly shows that the number of quantization levels has a vital impact on the achievable diversity order of the considered system. In particular, when the number of quantization levels $L$ is greater than the PSK constellation order $M,$ the system SEP will decrease quickly to zero as the SNR goes to infinity; on the other hand, when $L<M,$ there will be a SEP floor, i.e., the system SEP will not decrease to zero even when the SNR tends to infinity. The interesting case is $L=M$ where the SEP of the system will decrease to zero as the SNR tends to infinity but with a slower rate compared to the case where $L>M.$ 
We also mention here a related work \cite{ADC},  which gave diversity order analysis for an uplink single-input-single-output (SISO) system with low-resolution analog-to-digital converters (ADCs). Interestingly, the diversity order results in \cite{ADC} are consistent with our results in Theorem \ref{diversity} for $N=1$, though they are derived for different systems and communication scenarios using completely different analysis tools.

In the rest part of this subsection, we will give an intuitive explanation of the diversity order results in Theorem \ref{diversity}. 
The discussions below are somewhat heuristic but shed useful light on why different diversity order results are obtained for the three different cases in Theorem \ref{diversity}. A rigorous proof of Theorem \ref{diversity} will be provided in the next subsection. 

 To begin, we give some intuitions on the relationship between the safety margin $\alpha$ in \eqref{alpha} and the system SEP.   It is obvious that in the noiseless case (i.e., $\rho=\infty$),  the intended symbol will be incorrectly decoded if and only if $\alpha\leq 0$, which corresponds to the case where the noiseless received signal lies outside the decision region of the intended symbol (see Fig. \ref{fig1}), i.e., $\text{SEP}=\mathbb{P}\left(\alpha\leq 0\right)$.  In the general case, roughly speaking, the safety margin describes how large the order of magnitude the additive noise is allowed to have such that the received signal still lies within the decision region, in which case an error will not occur.  
Note that as $\rho$ tends to infinity, the magnitude of the noise, i.e., $|n|$, is in the order of  $\rho^{-\frac{1}{2}}$. Then, the above discussions imply that the SEP is in the same order as $\mathbb{P}(\alpha\leq\rho^{-\frac{1}{2}})$. Therefore, characterizing the diversity order is equivalent to studying how fast the probability $\mathbb{P}(\alpha\leq\rho^{-\frac{1}{2}})$
 decreases as $\rho$ tends to infinity, i.e.,
 $$d=-\lim_{\rho\to+\infty}\frac{\log\mathbb{P}(\alpha\leq\rho^{-\frac{1}{2}})}{\log\rho}.$$
 The above equality will be implicitly shown in the proof of Theorem \ref{diversity} in the next subsection.
 
Next we focus on the estimation of $\mathbb{P}(\alpha\leq\rho^{-\frac{1}{2}})$. For the considered model, $\alpha$ in \eqref{alpha2} is a sum of $N$ i.i.d. random variables. We express $\alpha$ as 
 $\alpha=\frac{1}{\sqrt{N}}\sum_{i=1}^N\alpha_i$, where 
 \begin{equation}\label{alphai}
 \alpha_i=|h_i|v_i~\text{and }v_i=\cos\theta_i-|\sin\theta_i|\cot\frac{\pi}{M}
 \end{equation} with 
 $h_i\sim\mathcal{CN}(0,1)$ and $\theta_i$ uniformly distributed in $\left[-\frac{\pi}{L},\frac{\pi}{L}\right]$.
Note that when $L\geq M$, each $\alpha_i$ in \eqref{alphai} is nonnegative, and hence $\mathbb{P}(\alpha\leq\rho^{-\frac{1}{2}})$ is  in the same order as $[\mathbb{P}(\alpha_i\leq\rho^{-\frac{1}{2}})]^N$ (see Eqs. \eqref{I1} and \eqref{eq3} further ahead for the rigorous expressions) and
\begin{equation}\label{eq6}
\begin{aligned}
d&=-N\lim_{\rho\to+\infty}\frac{\log\mathbb{P}(\alpha_i\leq\rho^{-\frac{1}{2}})}{\log\rho}.\\
\end{aligned}
\end{equation} 
We next investigate the three cases in Theorem \ref{diversity}, i.e., $L>M,~ L=M, $ and $L< M$, separately. 

\textbf{Case 1: $\bm{L>M.}$} In this case, $v_i$ in \eqref{alphai} is bounded from below by a positive constant independent of $\rho$ and thus has no effect on estimating $\mathbb{P}\left(\alpha_i\leq \rho^{-\frac{1}{2}}\right)$ when $\rho$ tends to infinity. As such,  the distribution of $|h_i|$ will dominate in estimating  $\mathbb{P}\left(\alpha_i\leq \rho^{-\frac{1}{2}}\right)$, whose CDF is in the order of  $O(x^2)$ when $x$ is near $0$, making  $\mathbb{P}(\alpha_i\leq \rho^{-\frac{1}{2}})$ in the order of $\rho^{-1}$. Then, from \eqref{eq6}, we have $d=N$.

\textbf{Case 2: $\bm{L=M.}$} In this case, $v_i\geq 0$ and the probability of $v_i$ being very close to $0$ is nonzero. Consequently, either a small $v_i$ or a small $|h_i|$ can result in a small $\alpha_i$, in which case the CDF of $\alpha_i$  becomes in the order of $O(x)$ when $x$ is near 0, and thus $\mathbb{P}(\alpha_i\leq \rho^{-\frac{1}{2}})$ is in the order of $\rho^{-\frac{1}{2}}$ and $d=\frac{N}{2}$.

\textbf{Case 3: $\bm{L<M.}$} In this case, the probability of each $v_i$ being negative is nonzero, and thus the probability of $\alpha$  being negative is also nonzero, i.e.,  $\mathbb{P}\left(\alpha\leq 0\right)$  is  a constant independent of $\rho$. Hence, $d=0$.

\subsection{PROOF OF THEOREM \ref{diversity}}
In this subsection, we give the rigorous proof of Theorem \ref{diversity}.  
The key is to establish lower and upper bounds on the SEP of the following form by applying inequality \eqref{in4} and characterizing the distribution of $\alpha$ in \eqref{alpha2}:
\begin{equation}\label{outline}
l_{\text{low}}(\rho)^{-d}+o(\rho^{-d})\leq \text{SEP}\leq l_{\text{up}}(\rho)^{-d}+o(\rho^{-d}),
\end{equation}
where $l_{\text{low}}(\rho)$ and $l_{\text{up}}(\rho)$ are linear functions of $\rho$ 
 and $o(\rho^{-d})$ is a high-order infinitesimal of $\rho^{-d},$ i.e., $o(\rho^{-d})/\rho^{-d} \rightarrow 0$.  Then we can conclude from the definition in \eqref{def:diversity} and \eqref{outline} that the diversity order is $d$. The most nontrivial steps of the proof is to characterize the distribution of $\alpha$ in \eqref{alpha2}. Below is the detailed proof of Theorem \ref{diversity} in three separate cases.

\vspace{-0.5cm}
\subsubsection{Proof of the diversity order when ${L> M}$} 

We first consider the case where $L>M$. To begin, we derive upper and lower bounds on $\alpha$ in \eqref{alpha2}. Since the quantization error $\theta_i\in[-\frac{\pi}{L},\frac{\pi}{L}]$, we have 
\begin{equation*}
\begin{aligned}
\alpha&=\frac{1}{\sqrt{N}} \sum_{i=1}^N|h_i|\left(\cos\theta_i-\sin\theta_i \cot\frac{\pi}{M}\right)\\
&\geq\frac{1}{\sqrt{N}} \sum_{i=1}^N|h_i|\left(\cos\frac{\pi}{L}-\sin\frac{\pi}{L}\cot\frac{\pi}{M}\right)\\
&= \frac{c_0\|\h\|_1}{\sqrt{N}}\geq \frac{c_0\|\h\|_2}{\sqrt{N}},
\end{aligned}
\end{equation*}
where $c_0=\cos\frac{\pi}{L}-\sin\frac{\pi}{L}\cot\frac{\pi}{M}$ is a  positive constant when $L>M$. Similarly, we have $\alpha\leq \frac{\|\h\|_1}{\sqrt{N}}\leq \|\h\|_2.$
Using inequality \eqref{in4} and the above bounds on $\alpha$, we obtain 
\begin{subequations}
\begin{equation}\label{in22}
\text{SEP}\leq 2\mathbb{E}\left[Q\left(\frac{\sqrt{2}\sin\frac{\pi}{M}c_0\|\h\|_2}{\sqrt{N}\sigma}\right)\right]
\end{equation}
and 
\begin{equation}\label{in2}
\text{SEP}\geq\mathbb{E}\left[Q\left(\frac{\sqrt{2}\sin\frac{\pi}{M}\|\h\|_2}{\sigma}\right)\right].
\end{equation}
\end{subequations}
Note that $2\|\h\|_2^2$ is a chi-square random variable with $2N$ degrees of freedom, i.e., $2\|\h\|_2^2\sim \mathcal{X}^2_{2N}$. Hence, the moment generating function of $\|\h\|_2^2$ is 
$$\text{MGF}_{\|\h\|_2^2}(t)=\mathbb{E}\left[e^{t\|\h\|_2^2}\right]=\left(1-t\right)^{-N}.$$
Applying the well-known inequality $Q(x)\leq \frac{1}{2}e^{-\frac{1}{2}x^2}$, $x\geq0 $ \cite{digitalcommunication} to \eqref{in22}, we can upper bound the SEP as 
\begin{equation}\label{up1}
\begin{aligned}
\text{SEP}
&\leq \text{MGF}_{\|\h\|_2^2}\left(-\frac{\sin^2\frac{\pi}{M}c_0^2}{{N}\sigma^2}\right)\\
&=\left(1+\frac{\sin^2\frac{\pi}{M}c_0^2}{N}\rho\right)^{-N},
\end{aligned}
\end{equation}
where the last equality holds since $\rho=1/\sigma^2$.

On the other hand, by applying the Craig's representation of the $Q$-function \cite{craig}, i.e.,
\begin{equation}\label{Craig}
Q(x)=\frac{1}{\pi}\int_{0}^{\frac{\pi}{2}}e^{-\frac{x^2}{2\sin^2\theta}}d\theta,\quad x\geq 0,
\end{equation}  to \eqref{in2}, we can lower bound the SEP as 
\vspace{-0.1cm}\begin{equation*}\label{Craig2}
\begin{aligned}
\text{SEP}
&\geq\mathbb{E}\left[\frac{1}{\pi}\int_0^\frac{\pi}{2}e^{-\frac{\sin^2\frac{\pi}{M}\rho\|\h\|_2^2}{\sin^2\theta}}d\theta\right]\\
&=\frac{1}{\pi}\int_0^{\frac{\pi}{2}}\text{MGF}_{\|\h\|_2^2}\left(-\frac{\sin^2\frac{\pi}{M}\rho}{\sin^2\theta}\right)d\theta\\
&=\frac{1}{\pi}\int_0^{\frac{\pi}{2}}\left(1+\frac{\sin^2\frac{\pi}{M}}{\sin^2\theta}\rho\right)^{-N}d\theta,
\end{aligned}
\end{equation*}
where the first equality holds since the integrand is nonnegative and thus we are free to change the order of integral and expectation according to the Tonelli-Fubini Theorem (see, e.g., \cite{real}). 
Similar to \cite[Section III-A]{antenna}, we can further obtain the following lower bound on the SEP:
\begin{equation}\label{lb1}
\text{SEP}\geq\frac{1}{2\sqrt{\pi(N+\frac{1}{2})}}\left(1+\sin^2\frac{\pi}{M}\rho\right)^{-N}.
\end{equation}
It follows immediately from \eqref{up1} and \eqref{lb1} that $d=N$ in this case.
\vspace{-0.6cm}
 \subsubsection{Proof of the diversity order when $L=M$}
 As discussed in Case 2 in Section \ref{section4}-\ref{sec:4a}, when $L=M$, both $|h_i|$ and $v_i$, $i=1,2,\dots, N$, will play roles in the diversity order analysis, in which case the distribution of $\alpha$ needs to be investigated carefully. 
The following two lemmas give the PDF of  $\alpha_i$ and an upper bound on the PDF of $\alpha$ when $L=M$, respectively, which are important for our analysis.
\begin{lemma}\label{pdf}
When $L=M$, the \text{\normalfont{PDF}} of $\alpha_i$ is given by 
\begin{equation*}
 p_{\alpha_i}\hspace{-0.04cm}(x)\hspace{-0.06cm}=\hspace{-0.09cm}\left\{
 \begin{aligned}
\frac{2M\sin\frac{\pi}{M}}{\sqrt{\pi}}e^{-\sin^2\frac{\pi}{M}x^2}Q\left(\sqrt{2}\cos\frac{\pi}{M}x\right)\hspace{-0.08cm}, \hspace{0.1cm}\text{if }x&\geq0;\\
0,\hspace{5.1cm}\text{if }x&<0.
\end{aligned}
\right.
\end{equation*}

\end{lemma}

\begin{lemma}\label{pdf2}
When $L=M$, the {\normalfont{PDF}} of $\alpha$ can be upper bounded as 
$$p_{\alpha}(x)\leq \frac{M^N\sin\frac{\pi}{M}}{\sqrt{\pi}},\quad x\geq 0.$$
\end{lemma}
The proofs of Lemmas \ref{pdf} and \ref{pdf2} are given in  Appendices \ref{appendixA} and  \ref{appendixB}, respectively. Now we are ready to give the diversity order analysis for the case of $L= M$. 

We first prove that $d\geq\frac{N}{2}$ by giving an upper bound on the SEP.  Note that $\alpha\geq0$ when $L=M$. Then from \eqref{in4} and using the fact that $Q(x)\leq \frac{1}{2}e^{-\frac{1}{2}x^2}$ for $x\geq 0$, we have 
 \begin{equation}\label{eq1}
 \begin{aligned}
 \text{SEP}
&\leq \mathbb{E}_\alpha\left[e^{-\sin^2\frac{\pi}{M}\rho\alpha^2}\right]\\
&=\int_0^{+\infty}e^{-\sin^2\frac{\pi}{M}\rho x^2}p_{\alpha}(x)dx.\\
\end{aligned}
\end{equation}
For any given $\epsilon>0$, the integral in \eqref{eq1} can be split into the sum of two integrals as 
\begin{equation}
\begin{aligned}
&\quad~\int_0^{+\infty}e^{-\sin^2\frac{\pi}{M}\rho x^2}p_{\alpha}(x)dx\\
&=\int_0^{\rho^{-\frac{1-\epsilon}{2}}}\hspace{-0.4cm}e^{-\sin^2\frac{\pi}{M}\rho x^2}p_{\alpha}(x)dx\hspace{-0.01cm}+\hspace{-0.1cm}\int_{\rho^{-\frac{1-\epsilon}{2}}}
^{+\infty}e^{-\sin^2\frac{\pi}{M}\rho x^2}p_{\alpha}(x)dx\\
&\triangleq I_1+I_2.
\end{aligned}
\end{equation}
Next we give upper bounds on $I_1$ and $I_2$ separately. For $I_1$, recalling the relationship between $\alpha$ and $\alpha_i$ and noting that $e^{-\sin^2\frac{\pi}{M}\rho x^2}\leq 1$,  we get the following inequality:
\begin{equation}\label{I1}
\begin{aligned}
I_1\leq \int_0^{\rho^{-\frac{1-\epsilon}{2}}}\hspace{-0.3cm}p_\alpha(x) dx
&=\mathbb{P}\left(0\leq\alpha\leq \rho^{-\frac{1-\epsilon}{2}}\right)\\
&\leq \prod_{i=1}^N \mathbb{P}\left(0\leq\alpha_i\leq \sqrt{N}\rho^{-\frac{1-\epsilon}{2}}\right).
\end{aligned}
\end{equation}
Moreover, it follows from Lemma \ref{pdf} that 
$$
\begin{aligned}
&\quad~\mathbb{P}\left(0\leq\alpha_i\leq \sqrt{N}\rho^{-\frac{1-\epsilon}{2}}\right)\\
&=\frac{2M\sin\frac{\pi}{M}}{\sqrt{\pi}}\int_0^{\sqrt{N}\rho^{-\frac{1-\epsilon}{2}}}\hspace{-0.2cm}e^{-\sin^2\frac{\pi}{M}x^2}Q\left(\sqrt{2}\cos\frac{\pi}{M}x\right)dx\\
&\leq\frac{M\sqrt{N}\sin\frac{\pi}{M}}{\sqrt{\pi}}\rho^{-\frac{1-\epsilon}{2}},
\end{aligned}
$$
where the inequality is due to $e^{-\sin^2\frac{\pi}{M}x^2}Q\left(\sqrt{2}\cos\frac{\pi}{M}x\right)\leq \frac{1}{2}$ for $x\geq 0$.
Combining the above inequality with \eqref{I1}, we have
\begin{equation}\label{I1_2}
I_1\leq\left(\frac{M\sqrt{N}\sin\frac{\pi}{M}}{\sqrt{\pi}}\right)^N\rho^{-\frac{N}{2}(1-\epsilon)}.
\end{equation}
Now we upper bound $I_2$ by applying Lemma \ref{pdf2}:
\begin{equation}\label{I2}
\begin{aligned}
I_2&\leq\frac{M^N\sin\frac{\pi}{M}}{\sqrt{\pi}}\int_{\rho^{-\frac{1-\epsilon}{2}}}^{+\infty}e^{-\sin^2\frac{\pi}{M}\rho x^2}dx\\
&=\frac{M^N}{\sqrt{\rho}}Q\left(\sqrt{2}\sin\frac{\pi}{M}\rho^{\frac{\epsilon}{2}}\right)\\
&\leq \frac{M^N}{2\sqrt{\rho}}e^{-\sin^2\frac{\pi}{M}\rho^\epsilon}.
\end{aligned}
\end{equation}
It follows from \eqref{eq1}--\eqref{I2} that 
$$\text{SEP}\leq \left(\frac{M\sqrt{N}\sin\frac{\pi}{M}}{\sqrt{\pi}}\right)^N\rho^{-\frac{N}{2}(1-\epsilon)}+\frac{M^N}{2\sqrt{\rho}}e^{-\sin^2\frac{\pi}{M}\rho^\epsilon},$$
which implies that 
$d\geq\frac{N}{2}(1-\epsilon).$
Since the above inequality holds for any $\epsilon>0$,  we have $d\geq \frac{N}{2}.$

Next we give a lower bound on the SEP, which in turn gives an upper bound on  $d$. Applying inequality \eqref{in4} and the Craig's representation of the $Q$-function in \eqref{Craig}, we have 
\begin{equation}\label{L=M:1}
\begin{aligned}
\text{SEP}&\geq 
\frac{1}{\pi}\int_{0}^{\frac{\pi}{2}}\int_{0}^{+\infty}e^{-\frac{\sin^2\frac{\pi}{M}\rho x^2}{\sin^2\theta}}p_{\alpha}(x)dx d \theta,
\end{aligned}
\end{equation}
where we have changed the order of integral according to the Tonelli-Fubini Theorem. 
Focusing on the inner integral, we have
\begin{equation*}
\begin{aligned}
\int_{0}^{+\infty}e^{-\frac{\sin^2\frac{\pi}{M}\rho x^2}{\sin^2\theta}}p_{\alpha}(x) d x&\geq\int_{0}^{\rho^{-\frac{1}{2}}}e^{-\frac{\sin^2\frac{\pi}{M}\rho x^2}{\sin^2\theta}}p_{\alpha}(x)d x\\
&\geq\mathbb{P}\left(0\leq\alpha\leq \rho^{-\frac{1}{2}}\right)e^{-\frac{\sin^2\frac{\pi}{M}}{\sin^2\theta}},\\
\end{aligned}
\end{equation*}
which, together with \eqref{Craig} and \eqref{L=M:1},  further implies 
\begin{equation}\label{eq2}
\begin{aligned}
\text{SEP}
&\geq Q\left(\sqrt{2}\sin\frac{\pi}{M}\right)\mathbb{P}\left(0\leq\alpha\leq \rho^{-\frac{1}{2}}\right).
\end{aligned}
\end{equation}
The term $\mathbb{P}\left(0\leq\alpha\leq \rho^{-\frac{1}{2}}\right)$ in \eqref{eq2} can further be lower bounded as 
\begin{equation}\label{eq3}
\mathbb{P}\left(0\leq\alpha\leq \rho^{-\frac{1}{2}}\right)\geq\prod_{i=1}^N\mathbb{P}\left(0\leq\alpha_i\leq \frac{\rho^{-\frac{1}{2}}}{{\sqrt{N}}}\right).
\end{equation}
According to Lemma \ref{pdf}, \begin{equation*}
\begin{aligned}
&\quad~\mathbb{P}\left(0\leq\alpha_i\leq\frac{\rho^{-\frac{1}{2}}}{\sqrt{N}}\right)\\
&=\frac{2M\sin\frac{\pi}{M}}{\sqrt{\pi}}\int_0^{{(N\rho)^{-\frac{1}{2}}}}\hspace{-0.2cm}e^{-\sin^2\frac{\pi}{M}x^2}Q\left(\sqrt{2}\cos\frac{\pi}{M}x\right)dx.
\end{aligned}
\end{equation*}
When $\rho$ is sufficiently large, we have
$$e^{-\sin^2\frac{\pi}{M}x^2}Q\left(\sqrt{2}\cos\frac{\pi}{M}x\right)\geq \frac{1}{4},~~ x\in\left[0,(N\rho)^{-\frac{1}{2}}\right],$$
and thus 
\begin{equation}\label{eq4}
\mathbb{P}\left(0\leq\alpha_i\leq\frac{\rho^{-\frac{1}{2}}}{\sqrt{N}}\right)\geq\frac{M\sin\frac{\pi}{M}\rho^{-\frac{1}{2}}}{2\sqrt{N\pi}}.
\end{equation}
Combining \eqref{eq2}--\eqref{eq4} yields the following lower bound on the SEP:
 \begin{equation*}
 \text{SEP}\geq Q\left(\sqrt{2}\sin\frac{\pi}{M}\right)\left(\frac{M\sin\frac{\pi}{M}\rho^{-\frac{1}{2}}}{2\sqrt{N\pi}}\right)^N,
\end{equation*}
which further implies $d\leq\frac{N}{2}.$
In conclusion, we have $d=\frac{N}{2}$.

\vspace{-0.5cm}

\subsubsection{Proof of the diversity order when $L<M$}
Finally, we prove that $d=0$ when $L<M$.
Using inequality \eqref{in4} and the fact that $Q(x)\geq \frac{1}{2}$ for $x\leq 0$, we have
\begin{equation}\label{eq5}
\begin{aligned}
\text{SEP}&\geq \mathbb{E}_{\alpha}\left[Q\left(\frac{\sqrt{2}\sin\frac{\pi}{M}\alpha}{\sigma}\right)\right]\\
&\geq \frac{1}{2}\mathbb{P}\left(\alpha\leq 0\right)\geq \frac{1}{2}\prod_{i=1}^N\mathbb{P}\left(\alpha_i\leq 0\right)\\
&{\color{black}~\qquad\qquad\qquad= \frac{1}{2}\prod_{i=1}^N\mathbb{P}\left(v_i\leq 0\right),}
\end{aligned}
\end{equation}
{\color{black}where the last equality holds since $\alpha_i=|h_i|v_i$ (see \eqref{alphai}) and $|h_i|\geq 0$.}
{\color{black}Using similar arguments as in Lemma \ref{pdf_vi} in Appendix A, we can derive the CDF of $v_i$ for $L<M$:
$$F_{v_i}(x)=1-\frac{L}{M}+\frac{L\arcsin\left(x\sin\frac{\pi}{M}\right)}{\pi},~~\text{if }L<M,$$
 and hence
\begin{equation}\label{eq_less0}
\begin{aligned}
&\mathbb{P}(v_i\leq 0)=F_{v_i}(0)=1-\frac{L}{M}>0.
\end{aligned}
\end{equation}
}
\hspace{-0.15cm}Substituting the above inequality into \eqref{eq5} gives $$\text{SEP}\geq \frac{1}{2}\left(1-\frac{L}{M}\right)^N,$$ which, together with the definition in \eqref{def:diversity}, shows  
$$0\leq d \leq \lim_{\rho\to\infty}\frac{-\ln\frac{1}{2}(1-L/M)^N}{\ln \rho}=0,$$ i.e.,
$d=0.$

{\color{black}We remark here that \eqref{eq_less0} holds regardless of the channel distribution, and hence the above diversity order result for $L<M$ (i.e., $d=0$) holds for any generic channel.} 

\section{NUMERICAL RESULTS}\label{section5}
In this section, we provide simulation results to verify the diversity order results in Theorem \ref{diversity}. All results are averaged over $10^9$ channel realizations. 
\begin{figure}[t]
\includegraphics[width=6cm]{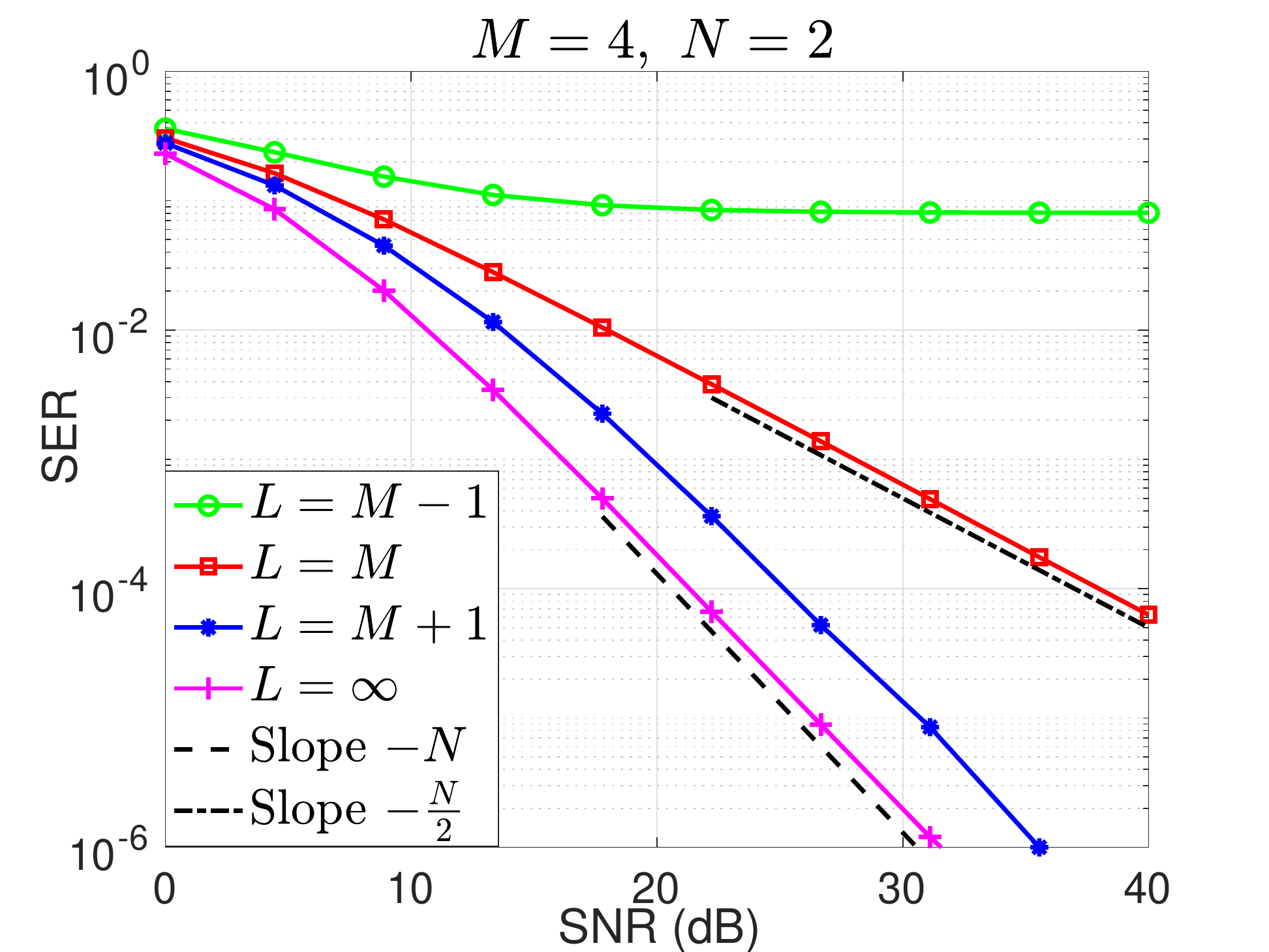}
\centering
\vspace{-0.05cm}
\caption{The SER versus the SNR, for different numbers of quantization levels $L$ with $M=4$ and $N=2$.}
\label{fig2}
\end{figure}
\begin{figure}[t]
\includegraphics[width=6cm]{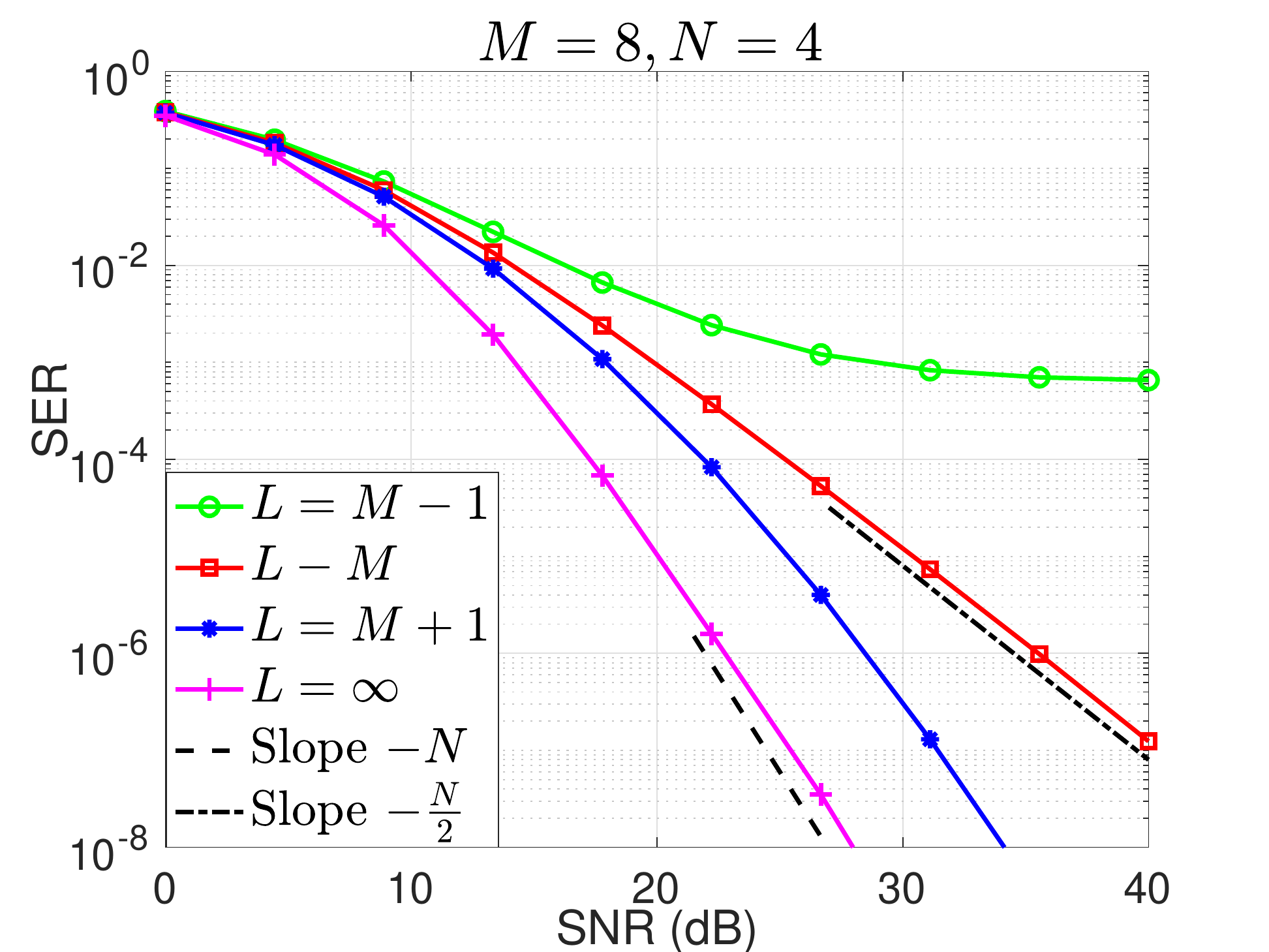}
\vspace{-0.05cm}
\centering
\caption{The SER versus the SNR, for different numbers of quantization levels of $L$ with $M=8$ and $N=4$.}
\label{fig3}
\end{figure}

In Fig. \ref{fig2} and Fig. \ref{fig3}, we consider  QPSK constellation with $N=2$ and 8-PSK constellation with $N=4$, respectively. We depict the SER as a function of the SNR and consider different numbers of  quantization levels, i.e., $L=M-1, L=M, L=M+1,$ and $L=\infty$, to demonstrate the effect of the number of quantization levels on the diversity order.  For clarity, we also report the lines with slopes $-N$ and  $-\frac{N}{2}$ to compare the simulation results with the analytical results. As shown in the figures, the curve of $L=M+1$ is parallel to the curve of $L=\infty$ in the high SNR regime, both of which are parallel to the line with slope $-N$; when $L=M$, the slope of the SER curve is nearly $-\frac{N}{2}$ when the SNR is high; and when $L<M$, there is an SER floor at high SNRs, i.e., the slope of the SER curve is $0$. These observations are consistent with the diversity order results in Theorem \ref{diversity}.

\section{CONCLUSION}\label{section6}
This paper characterized the diversity order of QCE transmission for a downlink single-user  MISO system  with $M$-PSK modulation. It has been shown that for the $L$-level quantized MF precoder, full diversity order is achievable when $L>M$, while only half and zero diversity order can be achieved when  $L=M$ and $L<M$, respectively. Simulation results verified our diversity order results. 

{\color{black}An important and interesting future work is to analyze the SEP performance of the multi-user QCE system. By simulations, we observe that  as the SNR grows, the SEP of the multi-user system does not decrease to zero and there is a positive SEP floor even when the number of quantization levels is infinite, which is in sharp contrast to the single-user case.
Therefore, the diversity order considered in this paper is no longer an appropriate performance metric in the multi-user scenario. Instead, it would be interesting to characterize the SEP floor at the infinite SNR for different numbers of quantization levels.  The SEP analysis of the multi-user system requires more sophisticated tools due to the more complicated system model. }
\appendices

\section{PROOF OF LEMMA \ref{pdf}}\label{appendixA}
Our goal in this section is to calculate the PDF of $\alpha_i=|h_i|v_i$ in \eqref{alphai} {\color{black}for the case where $L=M$}. Note that $2|h_i|^2\sim\mathcal{X}^2(2)$. Then the PDF of $|h_i|$ is given by $p_{|h_i|}(x)=2xe^{-x^2},~x\geq 0$. We next calculate the PDF of $v_i$.
\begin{lemma}\label{pdf_vi}
{\color{black}When $L=M$}, the PDF of $v_i$ is given by
\begin{equation*}\label{vi}
p_{v_i}(x)=\left\{
\begin{aligned}
\frac{M\sin\frac{\pi}{M}}{\pi\sqrt{1-\sin^2\frac{\pi}{M}x^2}},\quad &\text{if }x\in\left[0,1\right];\\
0,\hspace{1.3cm}~~~~~&\normalfont{\text{{otherwise}}}.
\end{aligned}\right.
\end{equation*}
\end{lemma}
\vspace{-0.3cm}
\begin{proof}
From the definition of $v_i$ in \eqref{alphai}, we have  
\begin{equation}\label{vi}
\begin{aligned}
v_i&=\cos\theta_i-|\sin\theta_i|\cot\frac{\pi}{M}\\
&=\cos|\theta_i|-\sin|\theta_i|\cot\frac{\pi}{M}\\
&=\frac{\sin\left(\frac{\pi}{M}-|\theta_i|\right)}{\sin\frac{\pi}{M}},
\end{aligned}
\end{equation}
where $|\theta_i|$ is uniformly distributed in $[0,\frac{\pi}{M}]$ and  the second equality holds since $|\theta_i|\leq\frac{\pi}{M}\leq \frac{\pi}{2}$. 
It follows immediately that  
 $0\leq v_i\leq 1.$
Therefore, if $x\notin[0,1]$,  $p_{v_i}(x)=0$; if 
$x\in\left[0,1\right]$,
 the CDF of $v_i$ is 
\begin{equation}\label{Fvi}
\begin{aligned}
F_{v_i}(x)&=\mathbb{P}\left(v_i\leq x\right)\\
&=\mathbb{P}\left(|\theta_i|\geq\frac{\pi}{M}-\arcsin\left(x\sin\frac{\pi}{M}\right)\right)\\
&=\frac{M\arcsin\left(x\sin\frac{\pi}{M}\right)}{\pi}.
\end{aligned}
\end{equation}
Taking the derivative with respect to $x$ on both sides of \eqref{Fvi} gives the desired result.
\end{proof}
The following lemma gives  the PDF of the product  of two independent random variables.
\begin{lemma}[{\cite[p135]{probability}}]\label{product}
Let $X$ and $Y$ be two independent random variables with PDFs $p_X(\cdot)$ and $p_Y(\cdot)$. Then the PDF of $Z=XY$ is
\begin{equation*}
p_Z(z)=\int_{-\infty}^{+\infty}\frac{1}{|x|}p_X(x)p_Y\left(\frac{z}{x}\right)dx.\end{equation*}
\end{lemma}
With the above two lemmas,  we are now ready to compute the PDF of $\alpha_i$.
Applying Lemma \ref{product} to $\alpha_i=v_i|h_i|$, we have
\begin{equation*}
\begin{aligned}
p_{\alpha_i}(x)
=&\left\{
\begin{aligned}
\int_0^1\frac{1}{z}p_{v_i}(z)p_{|h_i|}\left(\frac{x}{z}\right) dz, \qquad\text{if }x\geq 0;\\
0,\qquad\qquad\qquad \text{if }x<0.
\end{aligned}\right.
\end{aligned}
\end{equation*}
Now we simplify the above PDF expression for $x\geq 0$:
\begin{equation}\label{xgeq01}
\begin{aligned}
p_{\alpha_i}(x)=&~\frac{2M\sin\frac{\pi}{M}}{\pi}\int_{0}^1\frac{x}{z^2\sqrt{1-\sin^2\frac{\pi}{M}z^2}}e^{-\frac{x^2}{z^2}}dz\\
\xlongequal[]{z=\frac{\sin\theta}{\sin\frac{\pi}{M}}}&~\frac{2M\sin^2\frac{\pi}{M}}{\pi}\int_{0}^{\frac{\pi}{M}}\frac{x}{\sin^2\theta}e^{-\frac{\sin^2\frac{\pi}{M}x^2}{\sin^2\theta}}d\theta.
\end{aligned}
\end{equation}
Letting $u=\sin\frac{\pi}{M}x\cot\theta$, then 
\begin{equation*}\label{xgeq02}
\begin{aligned}
&\int_{0}^{\frac{\pi}{M}}\frac{x}{\sin^2\theta}e^{-\frac{\sin^2\frac{\pi}{M}x^2}{\sin^2\theta}}d\theta\\
=\hspace{0.08cm}&\frac{1}{\sin\frac{\pi}{M}}e^{-\sin^2\frac{\pi}{M}x^2}\int_{\cos\frac{\pi}{M}x}^{+\infty}e^{-u^2} du\\
=\hspace{0.08cm}&\hspace{-0.05cm}\frac{\sqrt{\pi}e^{-\sin^2\frac{\pi}{M}x^2}}{\sin\frac{\pi}{M}}Q\left(\sqrt{2}\cos\frac{\pi}{M}x\right)\\
\end{aligned}
\end{equation*}
Combining the above equality with \eqref{xgeq01} gives the PDF of $\alpha_i$ for $x\geq 0$ and 
completes the proof.
\section{PROOF OF LEMMA \ref{pdf2}}\label{appendixB}
In this section, we prove Lemma \ref{pdf2}. Let $S_n=\sum_{i=1}^n \alpha_i, ~n=1,2,\dots, N$. Then $\alpha=\frac{S_N}{\sqrt{N}}$.
{\color{black}
According to Lemma \ref{pdf}, the PDF of each $\alpha_i$ is zero when $x<0$, and thus $p_\alpha(x)$ is zero when $x<0$.
 When $x\geq0$, we claim that  the following inequality holds:
\begin{equation}\label{estimate}
p_{S_n}(x)\leq \frac{M^n\sin\frac{\pi}{M}}{\sqrt{n\pi}}e^{-\frac{1}{n}\sin^2\frac{\pi}{M}x^2},~~x\geq 0,
\end{equation}
which further implies that
$$p_{\alpha}(x)=\sqrt{N}p_{S_N}(\sqrt{N}x)\leq \frac{M^N\sin\frac{\pi}{M}}{\sqrt{\pi}},~~x\geq0.$$
We next prove  \eqref{estimate}  by induction.
Note that $Q(x)\leq \frac{1}{2}$ for all $x\geq 0$, and thus the inequality holds immediately for $n=1$. Now suppose that \eqref{estimate} holds for some $n$ with $1<n<N$. Then for $S_{n+1}$, we have
 \begin{equation*}
\begin{aligned}
&\quad~p_{S_{n+1}}(x)\\
&=\int_{-\infty}^{+\infty}p_{S_n}(y)p_{\alpha_{n+1}}(x-y) dy\\
&\leq\int_{0}^x \frac{M^n\sin\frac{\pi}{M}}{\sqrt{n\pi}}e^{-\frac{1}{n}\sin^2\frac{\pi}{M}y^2} \frac{M\sin\frac{\pi}{M}}{\sqrt{\pi}}e^{-\sin^2\frac{\pi}{M}(x-y)^2}dy\\
&\leq \frac{M^{n+1}\sin^2\frac{\pi}{M}}{\sqrt{n}\pi} \int_{-\infty}^{+\infty}e^{-\frac{1}{n}\sin^2\frac{\pi}{M}y^2} e^{-\sin^2\frac{\pi}{M}(x-y)^2}dy\\
&=\frac{M^{n+1}\sin\frac{\pi}{M}}{\sqrt{(n+1)\pi}}e^{-\frac{1}{n+1}\sin^2\frac{\pi}{M}x^2},
\end{aligned}
\end{equation*}
where  the first inequality uses \eqref{estimate} for $n$ and Lemma \ref{pdf} and the second inequality is due to the change of the integral interval. 
The above inequality shows that  \eqref{estimate} holds for $n+1$ and completes the proof.
\newpage


\begin{thebibliography}{10}
\providecommand{\url}[1]{#1}
\csname url@samestyle\endcsname
\providecommand{\newblock}{\relax}
\providecommand{\bibinfo}[2]{#2}
\providecommand{\BIBentrySTDinterwordspacing}{\spaceskip=0pt\relax}
\providecommand{\BIBentryALTinterwordstretchfactor}{4}
\providecommand{\BIBentryALTinterwordspacing}{\spaceskip=\fontdimen2\font plus
\BIBentryALTinterwordstretchfactor\fontdimen3\font minus
  \fontdimen4\font\relax}
\providecommand{\BIBforeignlanguage}[2]{{%
\expandafter\ifx\csname l@#1\endcsname\relax
\typeout{** WARNING: IEEEtran.bst: No hyphenation pattern has been}%
\typeout{** loaded for the language `#1'. Using the pattern for}%
\typeout{** the default language instead.}%
\else
\language=\csname l@#1\endcsname
\fi
#2}}
\providecommand{\BIBdecl}{\relax}
\BIBdecl

\bibitem{massivemimo2}
F.~Rusek, D.~Persson, B.~K. Lau, E.~G. Larsson, T.~L. Marzetta, O.~Edfors, and
  F.~Tufvesson, ``Scaling up {MIMO}: Opportunities and challenges with very
  large arrays,'' \emph{IEEE Signal Process. Mag.}, vol.~30, no.~1, pp. 40--60,
  Jan. 2013.

\bibitem{massivemimo1}
J.~G. Andrews, S.~Buzzi, W.~Choi, S.~V. Hanly, A.~Lozano, A.~C.~K. Soong, and
  J.~C. Zhang, ``What will {5G} be?'' \emph{IEEE J. Sel. Areas Commun.},
  vol.~32, no.~6, pp. 1065--1082, Jun. 2014.

\bibitem{massivemimo3}
L.~Lu, G.~Y. Li, A.~L. Swindlehurst, A.~Ashikhmin, and R.~Zhang, ``An overview
  of massive {MIMO}: Benefits and challenges,'' \emph{IEEE J. Sel. Topics
  Signal Process.}, vol.~8, no.~5, pp. 742--758, Oct. 2014.

\bibitem{pa}
O.~Blume, D.~Zeller, and U.~Barth, ``Approaches to energy efficient wireless
  access networks,'' in \emph{Proc. 4th Int. Symp. Commun. Control Signal
  Process.}, May 2010, pp. 1--5.

\bibitem{CE2}
S.~K. Mohammed and E.~G. Larsson, ``Single-user beamforming in large-scale
  {MISO} systems with per-antenna constant-envelope constraints: The doughnut
  channel,'' \emph{IEEE Trans. Wireless Commun.}, vol.~11, no.~11, pp.
  3992--4005, Nov. 2012.

\bibitem{CE}
S.~K. Mohammed and E.~G. Larsson, ``Per-antenna constant envelope precoding for large multi-user {MIMO}
  systems,'' \emph{IEEE Trans. Commun.}, vol.~61, no.~2, pp. 1059--1071, Mar.
  2013.

\bibitem{CEdesign1}
J.~Pan and W.-K. Ma, ``Constant envelope precoding for single-user large-scale
  {MISO} channels: Efficient precoding and optimal designs,'' \emph{IEEE J.
  Sel. Topics Signal Process.}, vol.~8, no.~5, pp. 982--995, Oct. 2014.

\bibitem{CEdesign2}
J.~Zhang, Y.~Huang, J.~Wang, B.~Ottersten, and L.~Yang, ``Per-antenna constant
  envelope precoding and antenna subset selection: A geometric approach,''
  \emph{IEEE Trans. Signal Process.}, vol.~64, no.~23, pp. 6089--6104, Dec.
  2016.

\bibitem{CEdesign3}
P.~V. Amadori and C.~Masouros, ``Constant envelope precoding by interference
  exploitation in phase shift keying-modulated multiuser transmission,''
  \emph{IEEE Trans. Wireless Commun.}, vol.~16, no.~1, pp. 538--550, Jan. 2017.

\bibitem{CEdesign4}
F.~Liu, C.~Masouros, P.~V. Amadori, and H.~Sun, ``An efficient manifold
  algorithm for constructive interference based constant envelope precoding,''
  \emph{IEEE Signal Process. Lett.}, vol.~24, no.~10, pp. 1542--1546, Oct.
  2017.

\bibitem{CEdesign5}
S.~Zhang, R.~Zhang, and T.~J. Lim, ``Constant envelope precoding for {MIMO}
  systems,'' \emph{IEEE Trans. Commun.}, vol.~66, no.~1, pp. 149--162, Jan.
  2018.

\bibitem{phy}
Z.~Wei, C.~Masouros, and F.~Liu, ``Secure directional modulation with few-bit
  phase shifters: Optimal and iterative-closed-form designs,'' \emph{IEEE
  Trans. Commun.}, vol.~69, no.~1, pp. 486--500, Jan. 2021.

\bibitem{radar}
S.~Ahmed, J.~S. Thompson, Y.~R. Petillot, and B.~Mulgrew, ``Finite alphabet
  constant-envelope waveform design for {MIMO} radar,'' \emph{IEEE Trans.
  Signal Process.}, vol.~59, no.~11, pp. 5326--5337, Nov. 2011.

\bibitem{irs}
H.~Yang, X.~Yuan, J.~Fang, and Y.-C. Liang, ``Reconfigurable intelligent
  surface aided constant-envelope wireless power transfer,'' \emph{IEEE Trans.
  Signal Process.}, vol.~69, pp. 1347--1361, Feb. 2021.

\bibitem{resolution1}
R.~Walden, ``Analog-to-digital converter survey and analysis,'' \emph{IEEE J.
  Sel. Areas Commun.}, vol.~17, no.~4, pp. 539--550, Apr. 1999.

\bibitem{trellis}
M.~Kazemi, H.~Aghaeinia, and T.~M. Duman, ``Discrete-phase constant envelope
  precoding for massive {MIMO} systems,'' \emph{IEEE Trans. Commun.}, vol.~65,
  no.~5, pp. 2011--2021, May 2017.

\bibitem{ciqce}
H.~Jedda, A.~Mezghani, A.~L. Swindlehurst, and J.~A. Nossek, ``Quantized
  constant envelope precoding with {PSK} and {QAM} signaling,'' \emph{IEEE
  Trans. Wireless Commun.}, vol.~17, no.~12, pp. 8022--8034, Dec. 2018.

\bibitem{AM}
J.-C. Chen, ``Efficient constant envelope precoding with quantized phases for
  massive {MU-MIMO} downlink systems,'' \emph{IEEE Trans. Veh. Technol.},
  vol.~68, no.~4, pp. 4059--4063, Apr. 2019.

\bibitem{GEMM}
M.~Shao, Q.~Li, W.-K. Ma, and A.~M.-C. So, ``A framework for one-bit and
  constant-envelope precoding over multiuser massive {MISO} channels,''
  \emph{IEEE Trans. Signal Process.}, vol.~67, no.~20, pp. 5309--5324, Oct.
  2019.

\bibitem{IDE}
C.-J. Wang, C.-K. Wen, S.~Jin, and S.-H. Tsai, ``Finite-alphabet precoding for
  massive {MU-MIMO} with low-resolution {DACs},'' \emph{IEEE Trans. Wireless
  Commun.}, vol.~17, no.~7, pp. 4706--4720, Jul. 2018.

\bibitem{SQUID}
S.~Jacobsson, G.~Durisi, M.~Coldrey, T.~Goldstein, and C.~Studer, ``Quantized
  precoding for massive {MU-MIMO},'' \emph{IEEE Trans. Commun.}, vol.~65,
  no.~11, pp. 4670--4684, Nov. 2017.

\bibitem{C3PO2}
O.~Castañeda, S.~Jacobsson, G.~Durisi, M.~Coldrey, T.~Goldstein, and
  C.~Studer, ``1-bit massive {MU-MIMO} precoding in {VLSI},'' \emph{IEEE J.
  Emerg. Sel. Topics Circuits Syst.}, vol.~7, no.~4, pp. 508--522, Dec. 2017.

\bibitem{SEP2}
F.~Sohrabi, Y.-F. Liu, and W.~Yu, ``One-bit precoding and constellation range
  design for massive {MIMO} with {QAM} signaling,'' \emph{IEEE J. Sel. Topics
  Signal Process.}, vol.~12, no.~3, pp. 557--570, Jun. 2018.

\bibitem{CIfirst}
H.~Jedda, A.~Mezghani, J.~A. Nossek, and A.~L. Swindlehurst, ``Massive {MIMO}
  downlink 1-bit precoding with linear programming for {PSK} signaling,'' in
  \emph{Proc. IEEE Workshop Signal Process. Adv. Wireless Commun.}, Jul. 2017,
  pp. 1--5.

\bibitem{CImodel}
A.~Li, C.~Masouros, F.~Liu, and A.~L. Swindlehurst, ``Massive {MIMO} 1-bit
  {DAC} transmission: A low-complexity symbol scaling approach,'' \emph{IEEE
  Trans. Wireless Commun.}, vol.~17, no.~11, pp. 7559--7575, Nov. 2018.

\bibitem{PBB}
A.~Li, F.~Liu, C.~Masouros, Y.~Li, and B.~Vucetic, ``Interference exploitation
  1-bit massive {MIMO} precoding: A partial branch-and-bound solution with
  near-optimal performance,'' \emph{IEEE Trans. Wireless Commun.}, vol.~19,
  no.~5, pp. 3474--3489, May 2020.

\bibitem{conference}
Z.~Wu, B.~Jiang, Y.-F. Liu, and Y.-H. Dai, ``A novel negative $\ell_1$ penalty
  approach for multiuser one-bit massive {MIMO} downlink with {PSK}
  signaling,'' in \emph{Proc. IEEE Int. Conf. Acoust., Speech, Signal
  Process.}, May 2022, pp. 5323--5327.

\bibitem{journal}
\BIBentryALTinterwordspacing
Z.~Wu, B.~Jiang, Y.-F. Liu, and Y.-H. Dai, ``{CI}-based one-bit precoding for multiuser downlink massive {MIMO}
  systems with {PSK} modulation: A negative $\ell_1$ penalty approach,'' 2021.
  [Online]. Available: \url{https://arxiv.org/abs/2110.11628}
\BIBentrySTDinterwordspacing

\bibitem{MFrate}
Y.~Li, C.~Tao, A.~Lee~Swindlehurst, A.~Mezghani, and L.~Liu, ``Downlink
  achievable rate analysis in massive {MIMO} systems with one-bit {DAC}s,''
  \emph{IEEE Commun. Lett.}, vol.~21, no.~7, pp. 1669--1672, Jul. 2017.

\bibitem{frequency}
S.~Jacobsson, G.~Durisi, M.~Coldrey, and C.~Studer, ``Linear precoding with
  low-resolution {DAC}s for massive {MU-MIMO-OFDM} downlink,'' \emph{IEEE
  Trans. Wireless Commun.}, vol.~18, no.~3, pp. 1595--1609, Mar. 2019.

\bibitem{ZF}
A.~K. Saxena, I.~Fijalkow, and A.~L. Swindlehurst, ``Analysis of one-bit
  quantized precoding for the multiuser massive {MIMO} downlink,'' \emph{IEEE
  Trans. Signal Process.}, vol.~65, no.~17, pp. 4624--4634, Sept. 2017.

\bibitem{CI3}
C.~Masouros, T.~Ratnarajah, M.~Sellathurai, C.~B. Papadias, and A.~K. Shukla,
  ``Known interference in the cellular downlink: {A} performance limiting
  factor or a source of green signal power?'' \emph{IEEE Commun. Mag.},
  vol.~51, no.~10, pp. 162--171, Oct. 2013.

\bibitem{CI2}
C.~Masouros, M.~Sellathurai, and T.~Ratnarajah, ``Vector perturbation based on
  symbol scaling for limited feedback {MISO} downlinks,'' \emph{IEEE Trans.
  Signal Process.}, vol.~62, no.~3, pp. 562--571, Feb. 2014.

\bibitem{CItutorial}
A.~Li, D.~Spano, J.~Krivochiza, S.~Domouchtsidis, C.~G. Tsinos, C.~Masouros,
  S.~Chatzinotas, Y.~Li, B.~Vucetic, and B.~Ottersten, ``A tutorial on
  interference exploitation via symbol-level precoding: Overview,
  state-of-the-art and future directions,'' \emph{IEEE Commun. Surveys Tuts.},
  vol.~22, no.~2, pp. 796--839, 2nd Quart. 2020.

\bibitem{diversity}
L.~Zheng and D.~Tse, ``Diversity and multiplexing: a fundamental tradeoff in
  multiple-antenna channels,'' \emph{IEEE Trans. Inf. Theory}, vol.~49, no.~5,
  pp. 1073--1096, May 2003.

\bibitem{SEP3}
M.~Shao, Q.~Li, Y.~Liu, and W.-K. Ma, ``Multiuser one-bit massive {MIMO}
  precoding under {MPSK} signaling,'' in \emph{Proc. IEEE Global Conf. Signal
  Inf. Process.}, Nov. 2018, pp. 833--837.

\bibitem{digitalcommunication}
J.~G. Proakis and M.~Salehi, \emph{Digital Communications}, 5th~ed.\hskip 1em
  plus 0.5em minus 0.4em\relax New York, NY, USA: McGraw-Hill, 2008.

\bibitem{ADC}
S.~Gayan, R.~Senanayake, H.~Inaltekin, and J.~Evans, ``Low-resolution
  quantization in phase modulated systems: Optimum detectors and error rate
  analysis,'' \emph{IEEE Open J. Commun. Soc.}, vol.~1, pp. 1000--1021, Jul.
  2020.

\bibitem{craig}
J.~Craig, ``A new, simple, and exact result for calculating the probability of
  error for two-dimensional signal constellations,'' in \emph{Proc. IEEE
  Military Commun. Conf.}, Nov. 1991, pp. 571--575.

\bibitem{real}
R.~M. Dudley, \emph{Real Analysis and Probability}, 2nd~ed.\hskip 1em plus
  0.5em minus 0.4em\relax Cambridge, UK: Cambridge University Press, 2002.

\bibitem{antenna}
R.~Jiang and Y.-F. Liu, ``Antenna efficiency in massive {MIMO} detection,'' in
  \emph{Proc. IEEE Workshop Signal Process. Adv. Wireless Commun.}, Sept. 2021,
  pp. 51--55.

\bibitem{probability}
V.~Rohatgi and A.~Saleh, \emph{An Introduction to Probability and Statistics},
  3rd~ed.\hskip 1em plus 0.5em minus 0.4em\relax Hoboken, NJ, USA: John Wiley
  \& Sons, 2015.

\end{thebibliography}

\end{document}